\documentclass{article}

\usepackage{amsmath,amsfonts,amsthm}
\usepackage{hyperref}
\usepackage{bm}
\usepackage{xcolor}

\usepackage{tikz}
\tikzstyle{nodestyle}=[shape=circle,shading=ball,ball color=red,draw]
\tikzstyle{edgestyle}=[thick,color=black!50!green,draw]
\tikzstyle{line a}=[draw, line width=0.2mm, color=black]
\tikzstyle{line b}=[draw, line width=0.2mm, color=red   , dash pattern=on 6 off 2]
\tikzstyle{line c}=[draw, line width=0.2mm, color=blue  , dash pattern=on 4 off 2]
\tikzstyle{line d}=[draw, line width=0.2mm, color=green , dash pattern=on 2 off 2]
\tikzstyle{line e}=[draw, line width=0.2mm, color=yellow, dash pattern=on 4 off 2 on 1 off 2]
\tikzstyle{line f}=[draw, line width=0.2mm, color=pink  , dash pattern=on 1 off 1 on 2 off 1]
\tikzstyle{line g}=[draw, line width=0.2mm, color=red!30]
\tikzstyle{xytics}=[draw, line width=0.2mm, color=black]
\tikzstyle{xtics label}=[anchor=north, inner sep=3, color=black]
\tikzstyle{ytics label}=[anchor=east, inner sep=3, color=black]
\tikzstyle{legend label}=[anchor=west, inner sep=5, color=black]
\tikzstyle{bars}=[draw, line width=0.2mm, color=black]
\tikzstyle{filledbars}=[draw, line width=0.2mm, color=black, fill=black!30]

\newcommand{\realspace}{\mathbb{R}}			
\newcommand{\vecspace}[2]{#1^{#2}}			
\newcommand{\rvecspace}[1]{\vecspace{\realspace}{#1}}	
\newcommand{\vecnorm}[1]{\left\|#1\right\|}		
\newcommand{\mattran}[1]{#1^{\text{T}}}			
\newcommand{\vecone}{{\bm 1}}			
\newcommand{\veczero}{{\bm 0}}			

\DeclareMathOperator*{\argmin}{\arg\min}

\DeclareMathOperator*{\tr}{tr}

\DeclareMathOperator*{\spspan}{span}

\newtheorem{thm}{Theorem}
\newtheorem{prop}{Proposition}
\newtheorem{dfn}{Definition}

\begin{document}

\title{A Fast Successive Over-Relaxation Algorithm for Force-Directed Network Graph Drawing }
\author{WANG Yong-Xian$^*$ and WANG Zheng-Hua \\
National University of Defense Technology,  Changsha  410073, China \\
* Corresponding author (email: yxwang@nudt.edu.cn)}

\date{Cite this article as:
\textcolor{red}{Wang, Y. \& Wang, Z. Sci. China Inf. Sci. (2012) 55: 677. https://doi.org/10.1007/s11432-011-4208-9}
\\
The final publication is available at \url{http://link.springer.com} via \url{http://dx.doi.org/10.1007/s11432-011-4208-9}.}
\maketitle

\begin{abstract}
Force-directed approach is one of the most widely used methods in graph drawing research. There are two main problems with the traditional force-directed algorithms. First, there is no mature theory to ensure the convergence of iteration sequence used in the algorithm and further, it is hard to estimate the rate of convergence even if the convergence is satisfied. Second, the running time cost is increased intolerablely in drawing large- scale graphs, and therefore the advantages of the force-directed approach are limited in practice. This paper is focused on these problems and presents a sufficient condition for ensuring the convergence of iterations. We then develop a practical heuristic algorithm for speeding up the iteration in force-directed approach using a successive over-relaxation (SOR) strategy. The results of computational tests on the several benchmark graph datasets used widely in graph drawing research show that our algorithm can dramatically improve the performance of force-directed approach by decreasing both the number of iterations and running time, and is 1.5 times faster than the latter on average.
\\
\textbf{keywords:} graph drawing, graph layout, successive over-relaxation, force-directed algorithm
\end{abstract}

\section{Introduction}
\label{sec:fds-1}

The visualization of large-scale complex networks paves the way for direct and further researches on complex networks. Recently, the network visualization problems, especially the problem of how to lay out networks in the 2-D or even the 3-D space aesthetically and efficiently have been attracting more and more interest.

Let's briefly review the graph drawing or layout of large-scale net graph first. In the graph theory, networks are defined as graphs consisting of two type sets: vertex sets and edge sets (sets of relations between vertices), i.e. $G = (V, E)$, where $V = \{ 1, 2, \dots, n\}$ is vertex set and $E \subset V \times V$ denotes the corresponding edge set. Some type of graphs, like the trees or the forests, have simple structures, but the others can be complicated, such as the flowcharts, the graphic representation for the real networks like the internet and so on. Graph visualization is to represent graphs in a plane (or three-dimensional space) in the form of picture, with vertices drawn as points and edges as lines connected with a pair of vertices. A challenging problem of automatic layout of graph is how, or how efficiently, to draw a graph with good aesthetics, as well as good structural properties of the original abstract graph. 

This paper will show how to obtain the layout of undirected graphs in a plane (or higher-dimensional space) aesthetically and efficiently. Throughout this paper, we assume that in the layout, edges can only be drawn as line segments instead of arcs.

Extensive work on the layout of graphs has been carried out. The earlier work mainly studied the layout of some specific graph types. Refs.~\cite{Batini:1986:LAD,Tamassia:1988:AGD,Carpano:1980:ADH,Sugiyama:1981:MVU,Rowe:1987:BDG} focused on the flowchart and some graphs with hierarchical structure, where the main idea is to arrange the vertices as some regular structures, such as grid \cite{Batini:1986:LAD,Tamassia:1988:AGD}, cicle \cite{Carpano:1980:ADH}, and some parallel lines \cite{Rowe:1987:BDG,Sugiyama:1981:MVU}. The time complexity of the algorithms is $\Theta (\left| V \right|^2)$. 
Eades \emph{et al} \cite{Eades:1984:AHG} proposed a new graph layout algorithm for the VLSI problem. In their algorithms, each vertex was imagined as a steel ball and each edge as a spring connected with both balls; thus the whole network of balls formed a mechanical system. Given an initial placement of all steel balls, the free stretching or squeezing of springs adjusts all the balls’ position dynamically, leading the ball-spring network to a equilibrium state finally after a sufficiently long time. This algorithm has $\Theta (\left| E \right|)$ time complexity. The sparser the edges in the graph, the more efficient the algorithm. Kamada \emph{et al} \cite{Kamada:1988:ADN,Kamada:1989:ADG} generalized this spring model by introducing an “ideal-distance-between-two-vertices” idea into the model, and transformed the problem of graph layout into a problem of minimum stress of dynamic system.
Fruchterman \emph{et al} \cite{Fruchterman:1991:GDF} further regarded an undirected connected graph as a mechanics system. According to a thought similar to the spring model, they proposed the so-called force-directed majorization method. Davidson \emph{et al} \cite{Davidson:1989:DGN}, inspired by the very large-scale integrated circuit problem, developed another optimizing method to address the graph drawing problem. They considered the vertex distribution, the distance between vertices and target plane, edge lengths, and the crossover among edges in the design of the stress (object) function as different weights in simulating-annealing optimization according to different standards for aesthetics. The time complexity of algorithms of Kamada, Fruchterman, Davidson are all $\Theta (\left| E \right|)$.
Since then the force-directed majorization, which is trying to minimize a stress function, has become an active topic in the layout of graph. Cohen discussed the parameters in the force-directed majorization \cite{Cohen:1997:TDG}, and Kaufmann \emph{et al} gave a detailed review on the methods \cite{Kaufmann:2001:DGM}.
On the other hand, Leeuw \emph{et al} \cite{Leeuw:2000:GLT} and Gansner \emph{et al} \cite{Gansner:2004:GDS} found that in the mathematics force-directed majorization has a similar formula to the stress function in multidimensional scaling (MDS) problem, and based on this finding they exchanged algorithms between these two fields. In the classical MDS fields, Kruskal and Leeuw \emph{et al} \cite{Kruskal:1964:NMS,Kruskal:1964:MSO,Leeuw:1988:CMM} developed an iterative algorithm and studied its iterative convergence properties especially the iterative convergence speed under some conditions. Their fruitful results have laid the foundations for some popular layout algorithms. The force-directed majorization algorithm proposed by Kamada et al. used a “one-by-one” style to update each vertex’s position in the target plane. The classical MDS method instead used a “batch” mode to update all the vertices in each iteration. According to these differences, Gansner \emph{et al} \cite{Gansner:2004:GDS} considered the batch-update mode in the force-directed majorization, and made further speedup by using math library of linear algebra. In recent years, the focuses are on the following aspects: the design of the object stress function on the basis of aesthetics rules, the improvement of the placement of some ad-hoc complex network graph, and the speedup of the drawing of large-scale graphs. For example, Refs.~\cite{Dwyer:2007:IER,Dwyer:2008:CSM,Dwyer:2008:CGL} discussed a minimizing stress function with some extra restraints, Refs.~\cite{Huangjingwei:2000:YGX:CHN,Zhangqingguo:2006:YYC:CHN} employed a genetic algorithm scheme to solve the KK and / or FR problem, and Ref.~\cite{Zhangweiming:2008:JYZ:CHN} improved the layout of a special type of network graph whose vertex degree obeys the power-law distribution efficiently.

There are still some difficulties in using the force-directed majorization. The first one is the low performance. Typically when a graph has more than 1000 vertices, both the performance and the aesthetics would decline. The second difficulty lies in judging whether a given iterated sequence is convergent or not and if it is convergent, how to estimate the convergence speed theoretically. Although they can fulfill the conditions of ensuring the iterative convergence based on a series of hypothetical assumptions, it is hard to say these available algorithms have the same use value in practice. We will deal with these problem in this paper, and propose a new condition for ensuring the convergence of iteration for the force-directed layout theoretically and develop a heuristic method with more practical values based on successive over-relaxation technique to accelerate the computations.

The remainder of this paper is organized as follows. Sect.~\ref{sec:fds-2} gives a brief introduction to the force- directed layout. Sect.~\ref{sec:fds-3} describes a sufficiency condition which we propose based on the fixed point theory of non-linear system and an estimate to the convergence rate. Because this method cannot cover all the cases in applications, we develop a heuristic method based on successive over-relaxation technique from practical view to accelerate the iterations in force-directed layout. Sect.~\ref{sec:fds-4} discusses this method in detail. Sect.~\ref{sec:fds-5} gives some numerical experiments in the benchmark test data sets, as well as the effect on the performance of some parameters’ selection. Sect.~\ref{sec:fds-6} summarizes the work. The implementation of the algorithm in this paper and some extra data can be found in supplementary material available at \url{http://wang.yongxian.googlepages.com/graphdraw}.

\section{Force-directed Majorization}
\label{sec:fds-2}

We focus on how to layout or draw the graph $G$ in space $\rvecspace{d}$. Let $G=(V,E)$ be an undirected graph, and $V=\left\{ {1,2,\ldots ,n} \right\}$ be the set of vertices or nodes, and $E$ be the set of edges. Each vertex $i$ in $G$ can be represented as a vector in $\rvecspace{d}$, say, $x_i =\mattran{\left( {x_{i1},x_{i2},\ldots,x_{id} } \right)}$. We call matrix $X=\mattran{(\mattran{x_1},\mattran{x_2},\ldots ,\mattran{x_n} )}\in \rvecspace{n\times d}$ a placement of graph when the elements of $X$ take a specific set of values. In force-directed majorization a stress minimization procedure is used and the object stress function is defined as 
\begin{align}\label{eq1}
        \min f(X) &= \sum\limits_{1\le i<j\le n} {w_{ij} \big( {\left\| {x_i -x_j } \right\|-d_{ij} } \big)^2} 
\end{align}
where $\left\| \cdot \right\|$ is the norm in the Euclidean space $\rvecspace{d}$, $\left\| {x_i -x_j } \right\|$ is the real distance between vertex $i$ and vertex $j$ in the current placement $X$, and $d_{ij} $ is the predefined ideal distance between vertex $i$ and vertex $j$. The weights $w_{ij} \ge 0$ represent  the contribution to the total stress of vertex pair ($i$, $j$). By Eq.~\eqref{eq1}, only if every pair of vertices' real distance equals their ideal distance, does the object stress reach the minimum value 0. In graph layout applications drawing the graph in plane means $d=2$ and $d=3$ in a three-dimensional space. In this paper we do not assume any such specific $d$'s value.

To solve optimization problem \eqref{eq1}, we firstly introduce a definition of dominant function as follow.

\begin{dfn}[dominant Function]\label{dfn:2-1}
        Assuming a function $f:\rvecspace{n}\to \realspace$ and $g:\rvecspace{n}\to \realspace$, for all $x\in \rvecspace{n}$, we have $f(x)\ge g(x)$, we call $f$ is a dominant function of $g$.
\end{dfn}

We now give a group of functions of stress function defined in Eq.~\eqref{eq1} \cite{Gansner:2004:GDS}.

\begin{prop}\label{prop:2-1}
        Given any $X\in\rvecspace{n\times d}$ and $Y\in\rvecspace{n\times d}$, define $g:\rvecspace{n\times d}\times \rvecspace{n\times d}\to \mathbb{R}$ as
        \begin{align}\label{eq2}
                g(X,Y)=\tr(\mattran{X}L^WX)-2\tr(\mattran{X}L^YY)+C
        \end{align}
        where $\tr(\cdot )$ denotes the trace of a matrix (sum of diagonal elements of the matrix). $C=\sum\limits_{i<j} {w_{ij} } d_{ij}^2$ is a constant. $L^W=\left( {l_{ij}^W } \right)_{n\times n} $ and $L^Y=\left( {l_{ij}^Y } \right)_{n\times n} $ are both Laplacian matrices, and their elements are respectively defined as:
        \begin{align}
        l_{ij}^W &= \begin{cases}
                        -w_{ij},                        & i\not=j \\
                        -\sum_{k\not=i} l_{ik}^W,       & i=j
                \end{cases} \\
        l_{ij}^Y &= \begin{cases}
                        0, & i\not=j\text{ÇÒ}y_i = y_j \\
                        -\frac{w_{ij} d_{ij}}{\vecnorm{y_i - y_j}}, & i\not=j\text{ÇÒ}y_i \not= y_j \\
                        -\sum_{k\not=i} l_{ik}^Y, & i=j
                \end{cases}
        \end{align}
We can conclude that $g(\cdot ,Y)$ is a class of dominant functions of $f(\cdot)$ defined by Eq.~\eqref{eq1}, that's to say, $f(X)\leq g(X,Y)$; the equality $f(X)=g(X,X)$ is satisfied if and only if $X=Y$.
\end{prop}
\begin{proof}
        By expanding the stress function in Eq.~\eqref{eq1} and using Cauchy-Schwartz inequality, it is easy to get the results. The readers can refer to Ref.~\cite{Gansner:2004:GDS} for details of the proof.
\end{proof}

\begin{prop}\label{prop:2-2}
        For any given initial placement $X^{(0)}$, define a sequence of placements as
        \begin{gather}\label{eq3}
        \begin{split}
                 Y^{(k+1)} &= \argmin_Y g(X^{(k)},Y)    \\
                 X^{(k+1)} &= \argmin_X g(X,Y^{(k+1)})
        \end{split}
        \end{gather}
        The object stress function defined in Eq.~\eqref{eq1} is non-decreasing in the sequence of placements $\left\{ {X^{(k)}} \right\}$.
\end{prop}
\begin{proof}
        From Proposition \ref{prop:2-1} we have
        \begin{align*}
                f(X) &\le g(X,Y), \quad \forall X,Y \\
                f(X) &=   g(X,X), \quad \forall X
        \end{align*}
        Because of the uniqueness of minimum, we have $Y^{(k+1)}=X^{(k)}$ from Eq.~\eqref{eq3}, so
        \begin{align}\label{eq4}
                X^{(k+1)} &=\argmin_X g(X,X^{(k)})
        \end{align}
        and the following chain inequalities hold:
        \begin{align*}
                f(X^{(k+1)})\le g(X^{(k+1)}, X^{(k)}) \le g(X^{(k)},X^{(k)})=f(X^{(k)}) 
        \end{align*}
\end{proof}
The iterative update process in Eq.~\eqref{eq3} (or, equivalently, Eq.~\eqref{eq4}) is the main step of force-directed majorization. The placement $X^{(k+1)}$ in iteration $k+1$ can be derived by searching for the minimum value of dominant function $g(\cdot,X^{(k)})$. The key point of this method is how to solve Eq.~\eqref{eq4} efficiently, and this can be done by the following algorithm.

\begin{prop}\label{prop:2-3}
        The solution to Eq.~\eqref{eq4} can be derived by solving the $d$ linear algebra equations below
        \begin{align}\label{eq:2-5}
                L^W\left[ {X^{(k+1)}} \right]_j = L^{X^{(k)}}\left[ {X^{(k)}} \right]_j, \quad j=1,2,\ldots ,d
        \end{align}
        where subscript $j$ denotes the vector composed of the $j$-th column of correspondent matrix.
\end{prop}
\begin{proof}
        As $g(X,Y)$ is a quadratic form with respect to $X$, it has a unique minimum on its domain, moreover at the minimum point, we have
        \begin{align}\label{eq5}
                L^W X &= L^Y Y
        \end{align}
        showing that the solution of Eq.~\eqref{eq4} can be got by solving Eq.~\eqref{eq:2-5}, and this concludes the claim.
\end{proof}

In Eq.~\eqref{eq5}, as the coefficient matrix, $L^W$ is a Laplacian of weight matrix $W=\left( {w_{ij} } \right)_{n\times n} $ and thus it is a positive semidefinite matrix with rank $n-1$. Its null space is $\spspan\{a\cdot \vecone\}$ ($\vecone$ denotes a vector with all 1 components).
Eq.~\eqref{eq5} cannot be solved directly for the singularity of $L^W$, and in classical MDS the Moore-Penrose inverse matrix is used to denote the solution by $X=\left[ {L^W} \right]^+L^YY$. Computing the solution by this method will make the computations too expensive to apply.
Gansner \emph{et al} \cite{Gansner:2004:GDS} presented another method to overcome this problem. By always taking $x_1 =\veczero$ and removing the first row and first column of $L^W$, as well as the first row of vector $L^Y Y$, we get a new $(n-1) \times (n-1)$ linear algebra system. The matrix of new system is strictly diagonal dominant and hence positive definite. Some direct methods (such as Cholesky factorization) or iterated methods (such as conjugate gradient, Gauss-Seidel iteration, etc) can be used here efficiently. We outline the process of force-directed majorization in Algorithm 1.

\begin{center}
\tiny
\begin{tabular}{rl}
\hline
	\multicolumn{2}{c}{Algorithm~1: Force-Directed Majorization for Graph Layout}	\\
\hline
1	&	initialize the placement $X^{(0)}$;	\\
2	&	$k \gets 0$;		\\
3	&	\textbf{Repeat};		\\
4	&	\qquad	$X^{(k+1)} \gets H\left( {X^{(k)}} \right)$; ( by solving the Eq.~\eqref{eq4} )\\
	&	\qquad	\\
	&	\qquad	\\
	&	\qquad	\\
	&	\qquad	\\
5	&	\qquad	$k\gets k+1;$		\\
6	&	\textbf{until} some termination conditions were satisfied;	\\
\hline
\end{tabular}
\qquad
\begin{tabular}{rl}
\hline
	\multicolumn{2}{c}{Algorithm~2: Succesive Over-Relaxation Algorithm for Graph Layout}	\\
\hline
1	&	initialize the placement $X^{(0)}$;	\\
2	&	$k \gets 0$;		\\
3	&	\textbf{Repeat};		\\
4	&	\qquad	$X^{(k+1)} \gets H\left( {X^{(k)}} \right)$; ( by solving the Eq.~\eqref{eq4} )\\
5	&	\qquad	$\tilde {X}^{(k+1)} \gets (1+\omega )X^{(k+1)}-\omega X^{(k)}$;	\\
6	&	\qquad	\textbf{if} $f\left( {\tilde {X}^{(k+1)}} \right) \le f\left( {X^{(k+1)}} \right)$ \textbf{then}	\\
7	&	\qquad	\qquad	$X^{(k+1)} \gets \tilde {X}^{(k+1)}$;	\\
8	&	\qquad	\textbf{end if}			\\
9	&	\qquad	$k\gets k+1;$		\\
10	&	\textbf{until} some termination conditions were satisfied;	\\
\hline
\end{tabular}
\end{center}

\section{Convergence Condition for Force-directed Layout}
\label{sec:fds-3}

Although force-directed layout is widely applied, there are still some problems awaiting to be solved in the basic theory. For example, to ensure the convergence of iteration what are the conditions to be satisfied by the object function? And how to estimate the convergence rate?  Guo \cite{Guoanxue:2003:BDD:CHN} and Zhang \cite{Zhanghongzhi:2002:DDHS:CHN} have proved the convergence of iteration in some classes of object functions, independently. In their cases, the object function should be a second-order differentiable contraction map. We will present a new sufficient condition for the convergence and give an estimate for the rate of convergence in this section. 

For simplicity of formula and without loss of generality, we assume $d = 1$ and replace the letters in uppercase $X,Y\in\rvecspace{n\times d}$ (matrix) with the ones in lowercase $x,y\in \rvecspace{n}$ (vector).

\begin{dfn}\label{dfn:3-1}
        Given a function $H:\rvecspace{n}\to \rvecspace{n}$ and the initial point $x^{(0)}\in \rvecspace{n}$, a sequence of points $\left\{ {x^{(k)}} \right\}$ defined by
        \begin{align}\label{eq6}
                x^{(k+1)} &= H(x^{(k)})
        \end{align}
        is called iterated sequence of $H$, and $H$ is called iteration function.
\end{dfn}
\begin{dfn}\label{dfn:3-2}
        For given function $H:\rvecspace{n}\to \rvecspace{n}$, we call $x^\ast $ a point of attraction of iteration \eqref{eq6} if there is an open neighborhood $S(x^\ast ,\varepsilon )=\left\{ {x: \left\| {x-x^\ast } \right\|<\varepsilon } \right\}$ of $x^\ast $ such that, $\forall x^{(0)}\in S(x^\ast ,\varepsilon )$, the iterated sequence in \eqref{eq6} converges to $x^\ast $.
\end{dfn}
\begin{dfn}\label{dfn:3-3}
        Let $H:\rvecspace{n}\to \rvecspace{n}$. $x\in \rvecspace{n}$ be a fixed point of $H$ if $H(x)=x$.
\end{dfn}
\begin{thm}\label{thm:3-1}
        Suppose that $x^\ast $ is a fixed point of function $H:\rvecspace{n}\to \rvecspace{n}$ and $H$ is Frenchet-differentiable at $x^\ast$. Let $D$ denote the differential operator. If the spectral radius $\rho (D H(x^\ast )) < 1$, then 
        (i) $x^\ast $ is a point of attraction of iterated sequence \eqref{eq6}, and 
        (ii) the convergence rate of the sequence is $\rho (D H(x^\ast ))$. Furthermore the convergence is linear if $\rho (D H(x^\ast ))>0$.
\end{thm}
\begin{proof}
    For (i), please refer to Ref.~\cite{ostrowski:1966:SES} or Ref.~\cite{Guoanxue:2003:BDD:CHN}(Theorem~5). (ii) can be found in Ref.~\cite{Ortega:2000:ISN}(Chapter~9,10).
\end{proof}

Theorem~\ref{thm:3-1} presents a sufficient condition for iteratively solving a non-linear system and it is required that iteration function   has some perfect character. In force-directed layout, following this conclusion, we introduce the ``A-condition''. below.

\noindent \textbf{(A-condition):}
 Suppose that $x^\ast $ is a fixed point of iteration function \eqref{eq4} and $f$ is second-order differentiable at $x^\ast $, while $g$ has second derivatives at $(x^\ast ,x^\ast )$. We also assume that the minimum of Eq.~\eqref{eq4} exists and is unique in each iteration.

\begin{thm}\label{thm:3-2}
        When ``A-condition'' is satisfied, force-directed layout algorithm has a convergence rate $1-\lambda _n (x^\ast )$ at $x^\ast $, where $\lambda _n (x^\ast )$ is the minimal eigenvalue of the generalized eigen-system 
        \begin{align}\label{eq7}
                D^2f(x^\ast )z &= \lambda D_{11} g(x^\ast ,x^\ast )z
        \end{align}
        where $D^2$ and $D_{11} $ denote second-order derivation operator and second-order partial derivative operator, respectively.
\end{thm}
\begin{proof}
        By Eq.~\eqref{eq4} an implicit iteration function $H: x^{(k+1)}=H\left( {x^{(k)}} \right)$ is defined, and under ``A-condition'' $x^\ast $ is a fixed point of $H$. From Eq.~\eqref{eq4} it follows that
        \begin{gather}
                D_1 g(x^{(k+1)},x^{(k)}) = \veczero     \nonumber \\
                {H}'(x) = -\left[ {D_{11} g(x,x)} \right]^{-1}D_{12} g(x,x)     \label{eq8}
        \end{gather}
        Notice that $f(x)=g(x,x)$. Then 
        \begin{align}\label{eq9}
                D^2f(x) &= D_{11} g(x,x)+D_{12} g(x,x)
        \end{align}
        Combining Eq.~\eqref{eq8} with \eqref{eq9}, we have
        \begin{align*}
                {H}'(x) &= I-[D_{11} g(x,x)]^{-1}D^2f(x)
        \end{align*}
        From Theorem~\ref{thm:3-1} we can conclude that the iteration in force-directed layout is convergent and the convergence rate is given by spectral radius $\rho ({H}'(x^\ast ))$ which is the maximal eigenvalue of matrix ${H}'(x^\ast )$, or equivalently, the minimal eigenvalue of generalized eigen-system \eqref{eq7}.
\end{proof}

In Theorem~\ref{thm:3-2} a sufficient condition, ``A-condition'', is given and the corresponding convergence rate is estimated. Compared with the algorithm in Refs.~\cite{Guoanxue:2003:BDD:CHN,Zhanghongzhi:2002:DDHS:CHN} the condition proposed here is not imposed on iteration function itself but on stress function and its dominant function, thus making it more fit for the implicit iteration function cases common in applications.

On the other hand ``A-condition'' make rigorous demands on the existence of fixed point, the differentiability of both stress function and its dominant function at the fixed point. This leads to many difficulties. Therefore, Theorem~\ref{thm:3-2} has more theoretical meaning than practical value. We will propose a new method in the Sect.~\ref{sec:fds-4} from the practical point of view to accelerate the iteration process in the force-directed layout.

\section{Accelerating Force-Directed Layout Based on Successive Over-Relaxation}
\label{sec:fds-4}

One of the difficulties of iterating method is that it is hard to pre-estimate the mounts of computations. The low rate of convergence and slow speed often make a loss of practicality even if the iteration is convergent. So we try to seek some accelerating methods. Correction and relaxation are two common techniques in numerical computation field \cite{Wangnengchao:2004:JSF:CHN}. In correction technique by making a small adjustment to current point $X^{(k)}$, a new corrected point $\tilde {X}^{(k)}$ is derived and it is closer to the true solution. In relaxation technique the new solution with higher approximate precision is derived by combining two old iterated solutions $\tilde {X}_1$ and $\tilde {X}_2$. We introduce the correction and relaxation techniques into the iteration relation $X^{(k+1)}=H\left( {X^{(k)}} \right)$ mentioned in Sect.~\ref{sec:fds-2}, and combine the current solution with the one in the next iteration into
\begin{align}\label{eq10}
        \tilde{X}^{(k+1)} &= (1+\omega )X^{(k+1)}-\omega X^{(k)}
\end{align}
where $\omega \ge 0$ is the relaxation factor. The intuitional meaning of Eq.~\eqref{eq10} can be explained as follows. Although $X^{(k+1)}$ and $X^{(k)}$ are both an approximation to the true solution $X^\ast $, the former is better than the latter for $f\left( {X^{(k+1)}} \right)\le f\left( {X^{(k)}} \right)$. In Eq.~\eqref{eq10} better $X^{(k+1)}$ has been strengthened while the worse $X^{(k)}$ has been suppressed in order to get an approximation better than these two. We call this method successive over-relaxation (SOR) and list the process as Algorithm~2. Apparently the origin algorithm is a special case of successive over-relaxation method when taking $\omega=0$. 

In Algorithm~2 we add an SOR accelerating step in lines 5--8 for comparison with Algorithm~1. In term of algorithm complexity, these new-added steps combine two existing solutions and the computation time can be ignored compared with solving linear system (line 4). This can also be confirmed in the numerical experiments in Sect.~\ref{sec:fds-5}. Lines 6--8 guarantee that combined solution cannot be deteriorated.

The SOR accelerating has another intuitional explanation. In classical force-directed layout, some branches of graph usually tend to move toward a specific direction in the layout plane across several iterations. This movement may be one among translation, rotation and reflection or their combinations (see supplementary materials available at \url{http://wang.yongxian.googlepages.com/graphdraw}). This phenomenon suggests that each vertex’s movement in successive iterations may not reach its ``saturation length” and thus has some potential to move a farther stride. Based on this observation we predict the next position of some vertices as
\begin{align}\label{eq11}
        \tilde {X}^{(k+1)} &= X^{(k+1)}+\omega \left( {X^{(k+1)}-X^{(k)}} \right)
\end{align}
This is just the equivalent form of Eq.~\eqref{eq10}. The parameter of stride factor $\omega $ is used to control the predicted position of the given vertex which locates on the extension line across both recent points $X^{(k+1)}$ and $X^{(k)}$. When $\omega=0$, SOR algorithm will degenerate into the original non-SOR algorithm.

\section{Experiments and Discussions}
\label{sec:fds-5}

\subsection{Data Sets}

To evaluate the performance of SOR acceleration, we choose some benchmark data commonly used in graph drawing as the data sets which cover all kinds of cases, such as a variety of vertex numbers, different degree distributions of vertices, weighted edges and unweighted edges, and so on. Two extra net examples come from real world (railway net of China and protein-protein interaction network of \emph{Cerevisiae Fermentum}) are also tested. Because the running time may be affected by the initial placement, we choose the initial layout randomly on a unit square in the plane and take every running time reported in this section as the arithmetic average of 20 runs. We take $w_{ij} = d_{ij}^{-2}$ in Eq.~\eqref{eq1}. The experiment platform is a PC with Pentium 4 CPU, 3.0GHz, 1.5GB main memory running Microsoft Windows XP (SP2) OS and the program written in Matlab script is used to implement all algorithms (programs and data set are available in supplementary material).

\subsection{How to Select Relax Factors}

In Algortihm 2, different relax factors will have much impact to the performance of algorithm. We adopt three strategies to evaluate the degree of this impact.

\subsubsection{Fixed Strategy}

In fixed strategy we determine the value of relax factor before iteration in Algorithm 2 and never change its value during all iterations. By choosing a variety of values for relax factor, we evaluate the relax factor's impact to the performance and the results on the data set \texttt{grid-1158} (including 1158 vertices) are reported in Fig.~\ref{fig:fixed-relax-factor}. The curves of iteration vs. stress in each of 80 iterations plotted in Fig.~\ref{fig:fixed-relax-factor}(a) and \ref{fig:fixed-relax-factor}(b) indicate that whatever relax factor, stress function decreases in SOR version faster than in non-SOR version. To make clear the relation between convergence rate and the degree of relaxation, the maximum number of iterations for reaching some given stress value in a variety of relaxation levels are examined in Fig.~\ref{fig:fixed-relax-factor}(c). Fig.~\ref{fig:fixed-relax-factor}(c). showed that SOR version needs less iterations than non-SOR version to reach the same stress value. More precisely, SOR method is only 40--60 percent or so of non-SOR method in running time. The bigger the relax factor's value taken, the more significant the difference is. However in the case of large relax factor, the ``invalid relaxation'' frequently occurs, suggesting that the condition in line 6 of Algorithm 2 is not satisfied. As a result the speedup of whole optimization process declines as shown in Fig.~\ref{fig:fixed-relax-factor}(d). The reason lies in the fact that only a few of relaxations are required in all iterations indeed in accelerating the optimization dramatically, and a large stride (i.e. relax factor)  farther than expected will lead to a sequence of conservative strides in successive iterations, thereby resulting in a decline in the total speedup. This explanation can also be confirmed in following tests.

\begin{figure}[h]
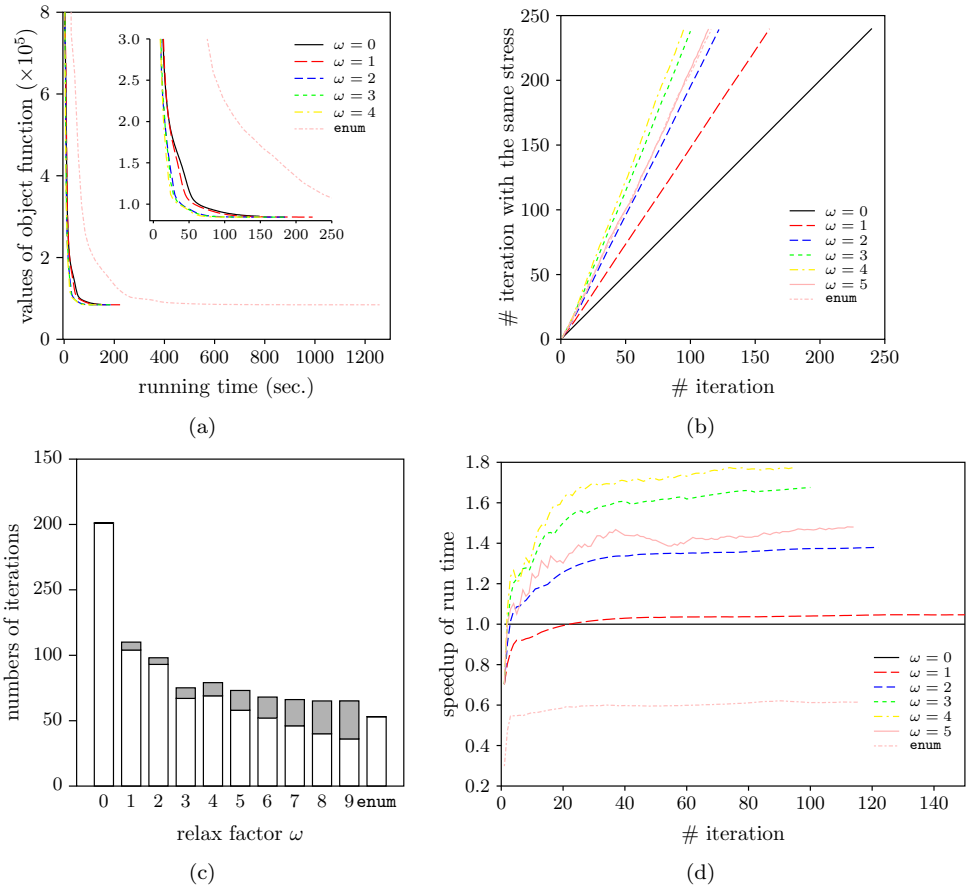

\centering
\begin{tabular}{cc}
        \scalebox{0.8}{\input{fixed-time-stress.pgf}}
&
        \scalebox{0.8}{\input{fig-iter-n-vs-1.pgf}}
\\
\footnotesize (a) &
\footnotesize (b) \\
        \scalebox{0.8}{\begin{tikzpicture}[x=0.847mm, y=1.131mm, inner xsep=0pt, inner ysep=0pt, outer xsep=0pt, outer ysep=0pt]

\path[xytics] (23.39,-12.56) -- (87.58,-12.56) -- (87.58,35.52) -- (23.39,35.52) -- cycle;

\node[outer ysep=6mm, anchor=north] at (55.485,-12.56) {relax factor $\omega$\strut};
\node[rotate=90] at (14-3,11.48) {numbers of iterations\strut};

\begin{scope}
\fontsize{9}{10}\selectfont
	\path[xytics] (28.73,-12.56) node[xtics label]{0\strut};
	\path[xytics] (34.10,-12.56) node[xtics label]{1\strut};
	\path[xytics] (39.44,-12.56) node[xtics label]{2\strut};
	\path[xytics] (44.78,-12.56) node[xtics label]{3\strut};
	\path[xytics] (50.13,-12.56) node[xtics label]{4\strut};
	\path[xytics] (55.49,-12.56) node[xtics label]{5\strut};
	\path[xytics] (60.84,-12.56) node[xtics label]{6\strut};
	\path[xytics] (66.18,-12.56) node[xtics label]{7\strut};
	\path[xytics] (71.53,-12.56) node[xtics label]{8\strut};
	\path[xytics] (76.89,-12.56) node[xtics label]{9\strut};
	\path[xytics] (82.23,-12.56) node[xtics label]{\texttt{enum}\strut};
	\path[xytics] (21.87,-12.56) node[ytics label]{  0\strut} -- (23.39,-12.56);
	\path[xytics] (21.87, -2.95) node[ytics label]{ 50\strut} -- (23.39,-2.95);
	\path[xytics] (21.87,  6.67) node[ytics label]{100\strut} -- (23.39,6.67);
	\path[xytics] (21.87, 16.30) node[ytics label]{250\strut} -- (23.39,16.30);
	\path[xytics] (21.87, 25.91) node[ytics label]{200\strut} -- (23.39,25.91);
	\path[xytics] (21.87, 35.52) node[ytics label]{150\strut} -- (23.39,35.52);
\end{scope}

\begin{scope}
	\path[bars] (26.86,-12.56) -- (30.62,-12.56) -- (30.62,26.12) -- (26.86,26.12) -- cycle;
	\path[filledbars] (26.86,26.11) -- (30.62,26.11) -- (30.62,26.12) -- (26.86,26.12) -- cycle;
	\path[bars] (32.21,-12.56) -- (35.98,-12.56) -- (35.98,7.46) -- (32.21,7.46) -- cycle;
	\path[filledbars] (32.21,7.44) -- (35.98,7.44) -- (35.98,8.61) -- (32.21,8.61) -- cycle;
	\path[bars] (37.57,-12.56) -- (41.33,-12.56) -- (41.33,5.34) -- (37.57,5.34) -- cycle;
	\path[filledbars] (37.57,5.33) -- (41.33,5.33) -- (41.33,6.31) -- (37.57,6.31) -- cycle;
	\path[bars] (42.92,-12.56) -- (46.67,-12.56) -- (46.67,0.35) -- (42.92,0.35) -- cycle;
	\path[filledbars] (42.92,0.33) -- (46.67,0.33) -- (46.67,1.89) -- (42.92,1.89) -- cycle;
	\path[bars] (48.26,-12.56) -- (52.02,-12.56) -- (52.02,0.72) -- (48.26,0.72) -- cycle;
	\path[filledbars] (48.26,0.71) -- (52.02,0.71) -- (52.02,2.65) -- (48.26,2.65) -- cycle;
	\path[bars] (53.60,-12.56) -- (57.38,-12.56) -- (57.38,-1.39) -- (53.60,-1.39) -- cycle;
	\path[filledbars] (53.60,-1.41) -- (57.38,-1.41) -- (57.38,1.50) -- (53.60,1.50) -- cycle;
	\path[bars] (58.97,-12.56) -- (62.72,-12.56) -- (62.72,-2.54) -- (58.97,-2.54) -- cycle;
	\path[filledbars] (58.97,-2.56) -- (62.72,-2.56) -- (62.72,0.53) -- (58.97,0.53) -- cycle;
	\path[bars] (64.31,-12.56) -- (68.07,-12.56) -- (68.07,-3.69) -- (64.31,-3.69) -- cycle;
	\path[filledbars] (64.31,-3.70) -- (68.07,-3.70) -- (68.07,0.16) -- (64.31,0.16) -- cycle;
	\path[bars] (69.66,-12.56) -- (73.41,-12.56) -- (73.41,-4.85) -- (69.66,-4.85) -- cycle;
	\path[filledbars] (69.66,-4.87) -- (73.41,-4.87) -- (73.41,-0.04) -- (69.66,-0.04) -- cycle;
	\path[bars] (75.00,-12.56) -- (78.77,-12.56) -- (78.77,-5.61) -- (75.00,-5.61) -- cycle;
	\path[filledbars] (75.00,-5.63) -- (78.77,-5.63) -- (78.77,-0.04) -- (75.00,-0.04) -- cycle;
	\path[bars] (80.36,-12.56) -- (84.12,-12.56) -- (84.12,-2.35) -- (80.36,-2.35) -- cycle;
	\path[filledbars] (80.36,-2.36) -- (84.12,-2.36) -- (84.12,-2.35) -- (80.36,-2.35) -- cycle;
\end{scope}

\end{tikzpicture}%
}
&
        \scalebox{0.8}{\input{fig-timespeedup-n-vs-1.pgf}}
\\
\footnotesize (c) &
\footnotesize (d)
\end{tabular}
\caption{Speedup of SOR algorithm}
\label{fig:fixed-relax-factor}
\end{figure}

\subsubsection{Enumerating Strategy}

In enumerating strategy we seek the best relax factor in $k$-th iteration by using an exhausting search method, i.e.:
\[\omega_\ast^{(k+1)} = \argmin_{\omega\in[0,+\infty)} f\left((1+\omega )X^{(k+1)}-\omega X^{(k)}\right).\]
For simplicity we search for the optimal relax factor in the discrete candidate set $\{0, 0.5, 1, 1.5, \dots, 8.5, 9\}$ in the tests. The results (Fig.~\ref{fig:fixed-relax-factor}(b) and \ref{fig:fixed-relax-factor}(c)) show that to reach the same stress value, the number of iteration in enumerating strategy should be further decreased to only 30--40 percents of that in non-SOR version. On the other hand, the total running time becomes longer for checking every candidate value of relax factor in each iteration and this makes the enumerating strategy unpractical.
As an example the relax factors of all iterations for drawing the graph \texttt{band-516} are illustrated in Fig.~\ref{fig:dynamic-relax-factor}. From the example we have two observations: (i) the optimal relax factors in different iterations are mainly located in a narrow interval [0.5,2], and (ii) a zigzag curve appears in the sequence of optimal relax factors; that is, if a big relax factor is taken in current iteration, the one taken in the next iteration should be smaller. This zigzag phenomenon is also seen in the fixed strategy previously. There is no obvious relation between optimal relax factor and the iteration.

\begin{figure}[h]
\centering
\scalebox{0.6}{\begin{tikzpicture}[x=2mm, y=1mm, inner xsep=0pt, inner ysep=0pt, outer xsep=0pt, outer ysep=0pt]

\path[xytics] (23.39,-12.56) -- (127.56,-12.56) -- (127.56,46.57) -- (23.39,46.57) -- cycle;

\node[inner ysep=7mm, anchor=north]  at (75.475,-12.56) {iteration\strut};
\node[rotate=90]  at (23.39-9.5+3,17) {best relax factor\strut};

\begin{scope}
\fontsize{9}{10}\selectfont
	\path[xytics] ( 23.39,-12.56) -- +(0, -1.11) node[xtics label]{  0\strut};
	\path[xytics] ( 37.09,-12.56) -- +(0, -1.11) node[xtics label]{ 50\strut};
	\path[xytics] ( 50.80,-12.56) -- +(0, -1.11) node[xtics label]{100\strut};
	\path[xytics] ( 64.50,-12.56) -- +(0, -1.11) node[xtics label]{150\strut};
	\path[xytics] ( 78.21,-12.56) -- +(0, -1.11) node[xtics label]{200\strut};
	\path[xytics] ( 91.93,-12.56) -- +(0, -1.11) node[xtics label]{250\strut};
	\path[xytics] (105.64,-12.56) -- +(0, -1.11) node[xtics label]{300\strut};
	\path[xytics] (119.34,-12.56) -- +(0, -1.11) node[xtics label]{350\strut};

	\path[xytics] (23.39,-12.56) -- +(-1.11, 0) node[ytics label]{0  \strut};
	\path[xytics] (23.39, -6.65) -- +(-1.11, 0) node[ytics label]{0.5\strut};
	\path[xytics] (23.39, -0.74) -- +(-1.11, 0) node[ytics label]{1.0\strut};
	\path[xytics] (23.39,  5.19) -- +(-1.11, 0) node[ytics label]{1.5\strut};
	\path[xytics] (23.39, 11.09) -- +(-1.11, 0) node[ytics label]{2.0\strut};
	\path[xytics] (23.39, 17.00) -- +(-1.11, 0) node[ytics label]{2.5\strut};
	\path[xytics] (23.39, 22.91) -- +(-1.11, 0) node[ytics label]{3.0\strut};
	\path[xytics] (23.39, 28.82) -- +(-1.11, 0) node[ytics label]{3.5\strut};
	\path[xytics] (23.39, 34.75) -- +(-1.11, 0) node[ytics label]{4.0\strut};
	\path[xytics] (23.39, 40.66) -- +(-1.11, 0) node[ytics label]{4.5\strut};
	\path[xytics] (23.39, 46.57) -- +(-1.11, 0) node[ytics label]{5.0\strut};
\end{scope}

\begin{scope}

\path[xytics,line width=0.1bp] (23.67,46.57) -| (23.94,15.82) -| (24.48,15.82) -- (24.48,4.00) -| (24.76,-0.74)  -- (25.03,-0.74) -- (25.03,-7.83) -| (25.31,1.62) -| (25.58,-5.47) -| (25.86,6.37) -| (26.12,-7.83) -| (26.41,13.46) -| (26.69,-7.83) -| (26.95,8.73) -| (27.23,-7.83) -| (27.50,8.73) -| (27.78,-7.83) -| (28.05,6.37) -| (28.33,-5.47) -| (28.59,1.62) -| (28.88,-5.47) -| (29.14,6.37) -| (29.42,-5.47) -| (29.69,4.00) -| (29.97,-5.47) -| (30.25,4.00) -| (30.52,-5.47) -| (30.80,6.37) -| (31.06,-5.47) -| (31.34,4.00) -| (31.61,-5.47) -| (31.89,6.37) -| (32.16,-5.47) -| (32.44,4.00) -| (32.70,-5.47) -| (32.98,6.37) -| (33.27,-5.47) -| (33.53,4.00) -| (33.81,-5.47) -| (34.08,6.37) -| (34.36,-5.47) -| (34.62,4.00) -| (34.91,-5.47) -| (35.17,6.37) -| (35.45,-5.47) -| (35.72,6.37) -| (36.00,-5.47) -| (36.27,4.00) -| (36.55,-5.47) -| (36.83,6.37) -| (37.09,-5.47) -| (37.38,4.00) -| (37.64,-5.47) -| (37.92,6.37) -| (38.19,-5.47) -| (38.47,4.00) -| (38.73,-5.47) -| (39.02,6.37) -| (39.28,-5.47) -| (39.56,6.37) -| (39.85,-5.47) -| (40.11,4.00) -| (40.39,-5.47) -| (40.66,6.37) -| (40.94,-5.47) -| (41.20,4.00) -| (41.49,-5.47) -| (41.75,8.73) -| (42.03,-5.47) -| (42.30,13.46) -| (42.58,-7.83) -| (42.84,11.09) -| (43.13,-7.83) -| (43.41,11.09) -| (43.67,-7.83) -| (43.96,22.91) -| (44.22,-7.83) -| (44.50,8.73) -| (44.77,-7.83) -| (45.05,13.46) -| (45.31,-7.83) -| (45.60,8.73) -| (45.86,-7.83) -| (46.14,18.19) -| (46.43,-7.83) -| (46.69,8.73) -| (46.97,-5.47) -| (47.24,1.62) -| (47.52,-3.11) -| (47.78,1.62) -| (48.07,-5.47) -| (48.33,8.73) -| (48.61,-7.83) -| (48.88,22.91) -| (49.16,-7.83) -| (49.44,18.19) -| (49.71,-5.47) -| (49.99,4.00) -| (50.25,-5.47) -| (50.54,8.73) -| (50.80,-5.47) -| (51.08,4.00) -| (51.35,-3.11) -| (51.63,4.00) -| (51.89,-3.11) -| (52.18,6.37) -| (52.44,-3.11) -| (52.72,6.37) -| (53.00,-3.11) -| (53.27,4.00) -| (53.55,-3.11) -| (53.82,-0.74) -| (54.10,-3.11) -| (54.36,1.62) -| (54.64,-5.47) -| (54.91,8.73) -| (55.19,-5.47) -| (55.46,4.00) -| (55.74,-3.11) -| (56.02,4.00) -| (56.29,-3.11) -| (56.57,6.37) -| (56.83,-0.74) -| (57.11,6.37) -| (57.38,-0.74) -| (57.93,-0.74) -- (57.93,-5.47) -| (58.21,-0.74) -| (58.47,-3.11) -| (58.75,-0.74) -| (59.30,-0.74) -| (59.85,-0.74) -- (59.85,1.62) -| (60.13,-3.11) -| (60.40,6.37) -| (60.68,-5.47) -| (60.94,11.09) -| (61.22,-5.47) -| (61.49,4.00) -| (61.77,-3.11) -| (62.04,4.00) -| (62.32,-3.11) -| (62.60,1.62) -| (62.86,-3.11) -| (63.15,-0.74) -| (63.41,-3.11) -| (63.69,1.62) -| (63.96,-3.11) -| (64.24,1.62) -| (64.50,-3.11) -| (64.79,1.62) -| (65.05,-3.11) -| (65.33,1.62) -| (65.60,-3.11) -| (65.88,1.62) -| (66.16,-3.11) -| (66.43,4.00) -| (66.71,-0.74) -| (67.26,-0.74) -| (67.80,-0.74) -- (67.80,-3.11) -| (68.07,1.62) -| (68.35,-3.11) -| (68.61,1.62) -| (68.90,-3.11) -| (69.18,4.00) -| (69.44,-0.74) -| (69.73,6.37) -| (69.99,-3.11) -| (70.27,-0.74) -| (70.54,-3.11) -| (70.82,-0.74) -| (71.08,-3.11) -| (71.37,1.62) -| (71.63,-3.11) -| (71.91,1.62) -| (72.18,-3.11) -| (72.46,4.00) -| (72.74,-3.11) -| (73.01,1.62) -| (73.29,-0.74) -| (73.55,1.62) -| (73.84,-0.74) -| (74.10,4.00) -| (74.38,-0.74) -| (74.93,-0.74) -| (75.19,-3.11) -| (75.48,4.00) -| (75.76,-3.11) -| (76.02,1.62) -| (76.31,-3.11) -| (76.57,-0.74) -| (76.85,1.62) -| (77.12,-3.11) -| (77.40,1.62) -| (77.66,-3.11) -| (77.95,4.00) -| (78.21,-3.11) -| (78.49,1.62) -| (78.77,-3.11) -| (79.04,1.62) -| (79.32,-3.11) -| (79.59,4.00) -| (79.87,-3.11) -| (80.13,1.62) -| (80.42,-3.11) -| (80.68,-0.74) -| (81.23,-0.74) -| (81.77,-0.74) -| (82.34,-0.74) -| (82.60,1.62) -| (82.88,-3.11) -| (83.15,6.37) -| (83.43,-3.11) -| (83.70,1.62) -| (83.98,-0.74) -| (84.24,1.62) -| (84.52,-0.74) -| (85.07,-0.74) -| (85.62,-0.74) -- (85.62,-3.11) -| (85.90,-0.74) -| (86.17,-3.11) -| (86.45,4.00) -| (86.71,-0.74) -| (87.26,-0.74) -| (87.81,-0.74) -| (88.09,1.62) -| (88.35,-3.11) -| (88.63,1.62) -| (88.92,-3.11) -| (89.18,1.62) -| (89.46,-3.11) -| (89.73,1.62) -| (90.01,-3.11) -| (90.28,6.37) -| (90.56,-3.11) -| (90.82,4.00) -| (91.10,-3.11) -| (91.37,1.62) -| (91.65,-3.11) -| (91.93,4.00) -| (92.20,-3.11) -| (92.48,6.37) -| (92.74,-3.11) -| (93.03,4.00) -| (93.29,-3.11) -| (93.57,-0.74) -| (93.84,-3.11) -| (94.12,4.00) -| (94.39,-3.11) -| (94.67,4.00) -| (94.93,-3.11) -| (95.21,4.00) -| (95.50,-3.11) -| (95.76,1.62) -| (96.04,-3.11) -| (96.31,4.00) -| (96.59,-3.11) -| (96.85,1.62) -| (97.14,-3.11) -| (97.40,4.00) -| (97.68,-3.11) -| (97.95,1.62) -| (98.23,-0.74) -| (98.78,-0.74) -- (98.78,1.62) -| (99.06,-3.11) -| (99.32,6.37) -| (99.61,-5.47) -| (99.87,8.73) -| (100.15,-5.47) -| (100.42,8.73) -| (100.70,-5.47) -| (100.96,8.73) -| (101.25,-5.47) -| (101.53,8.73) -| (101.79,-5.47) -| (102.08,8.73) -| (102.34,-5.47) -| (102.62,11.09) -| (102.89,-5.47) -| (103.17,8.73) -| (103.43,-5.47) -| (103.72,11.09) -| (103.98,-5.47) -| (104.26,11.09) -| (104.53,-5.47) -| (104.81,13.46) -| (105.09,-5.47) -| (105.36,13.46) -| (105.64,-3.11) -| (105.90,6.37) -| (106.19,-3.11) -| (106.45,-0.74) -| (106.73,-5.47) -| (107.00,-3.11) -| (107.28,-5.47) -| (107.54,-3.11) -| (107.83,-5.47) -| (108.37,-5.47) -- (108.37,-3.11) -| (108.65,-5.47) -| (108.92,-0.74) -| (109.20,-7.83) -| (109.47,6.37) -| (109.75,-7.83) -| (110.01,4.00) -| (110.30,-5.47) -| (110.56,1.62) -| (110.84,-5.47) -| (111.11,4.00) -| (111.39,-5.47) -| (111.67,4.00) -| (111.94,-5.47) -| (112.22,4.00) -| (112.48,-5.47) -| (112.76,4.00) -| (113.03,-5.47) -| (113.31,4.00) -| (113.58,-5.47) -| (113.86,4.00) -| (114.12,-5.47) -| (114.41,4.00) -| (114.69,-5.47) -| (114.95,4.00) -| (115.23,-5.47) -| (115.50,4.00) -| (115.78,-5.47) -| (116.05,4.00) -| (116.33,-5.47) -| (116.59,4.00) -| (116.87,-5.47) -| (117.14,4.00) -| (117.42,-5.47) -| (117.69,6.37) -| (117.97,-7.83) -| (118.25,18.19) -| (118.51,-7.83) -| (118.80,6.37) -| (119.06,-5.47) -| (119.34,1.62) -| (119.61,-3.11) -| (119.89,-0.74) -| (120.16,-3.11) -| (120.44,-0.74) -| (120.98,-0.74) -- (120.98,-3.11) -| (121.27,-0.74) -| (121.53,-3.11) -| (121.81,1.62) -| (122.08,-3.11) -| (122.36,-0.74) -| (122.62,-3.11) -| (122.91,-0.74) -| (123.45,-0.74) -- (123.45,-3.11) -| (123.72,-0.74) -| (124.00,-3.11) -| (124.27,1.62) -| (124.55,-5.47) -| (124.83,6.37) -| (125.09,-5.47) -| (125.38,1.62) -| (125.64,-3.11) -| (125.92,-0.74) -| (126.19,-3.11) -| (126.47,1.62) -| (126.73,-5.47) -| (127.02,6.37) -| (127.28,-5.47) -| (127.56,1.62);

\end{scope}
\end{tikzpicture}%
}
\caption{
Iteration vs. optimal relax factor in SOR method}
\label{fig:dynamic-relax-factor}
\end{figure}
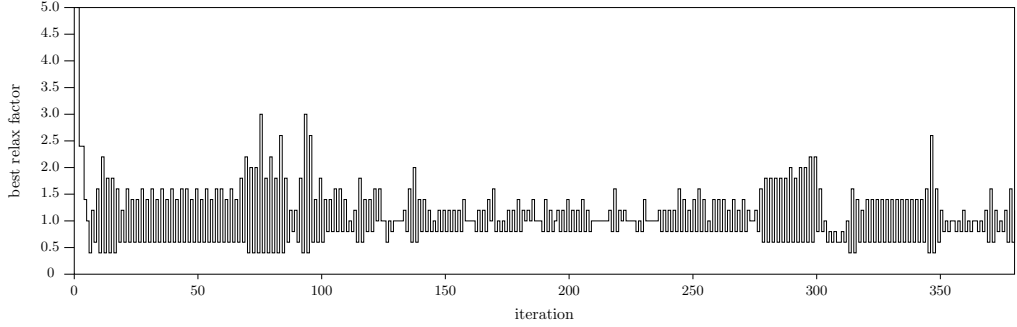

\subsubsection{Probabilistic Strategy}

The above discussion shows that fixed strategy can lead to ``invalid SOR'' cases, while the speedup is canceled out by extra computations introduced in enumerating strategy. Based on this observation, a tradeoff strategy can be taken, where a priori probabilistic distribution for candidate relax factors is used and a specific value in each iteration is selected by a roulette randomly. One extreme case in probabilistic strategy is that the same relax factor is selected in each iteration which degenerates into fixed strategy. And in another extreme case, every relax factor in each iteration is optimal, which is equivalent to the enumerating strategy. A comparison of performances for drawing graph \texttt{band-303} in three strategies is illustrated in Fig.~\ref{fig:relax-factor-det-method} where relax factors in every strategy fall in the interval [0.5,2]. From Fig.~\ref{fig:relax-factor-det-method} and Table~1 we could know that every iteration costs almost the same running time, and the number of iterations to reach the given stress is moderate compared with fixed strategy. The main advantage of probabilistic strategy lies in the fact that it avoids the blindness in choosing the relax factor in fixed strategy and decrease the time costs in enumerating strategy, thus giving rise to a good tradeoff between them.

\begin{figure}[h]
\centering
\begin{tabular}{cc}
        \scalebox{0.8}{\begin{tikzpicture}[x=0.8mm, y=1.1mm, inner xsep=0pt, inner ysep=0pt, outer xsep=0pt, outer ysep=0pt]

\path[xytics] (13.91,155.70) -- (89.45,155.70) -- (89.45,204.82) -- (13.91,204.82) -- cycle;

\node[outer ysep=6mm, anchor=north] at (51.68,155.70) {relax factor $\omega$\strut};
\node[rotate=90] at (4.4-3,180.26) {number of iterations\strut};

\begin{scope}
\fontsize{9}{10}\selectfont
	\path[xytics] (23.35,155.05) node[xtics label]{0\strut}   -- (23.35,155.70);
	\path[xytics] (32.79,155.05) node[xtics label]{0.5\strut} -- (32.79,155.70);
	\path[xytics] (42.23,155.05) node[xtics label]{1.0\strut} -- (42.23,155.70);
	\path[xytics] (51.67,155.05) node[xtics label]{1.5\strut} -- (51.67,155.70);
	\path[xytics] (61.11,155.05) node[xtics label]{2.0\strut} -- (61.11,155.70);
	\path[xytics] (70.55,155.05) node[xtics label]{\texttt{auto}\strut} -- (70.55,155.70);
	\path[xytics] (79.99,155.05) node[xtics label]{\texttt{enum}\strut} -- (79.99,155.70);

	\path[xytics] (12.39,155.70) node[ytics label]{  0\strut} -- (13.91,155.70);
	\path[xytics] (12.39,163.90) node[ytics label]{100\strut} -- (13.91,163.90);
	\path[xytics] (12.39,172.07) node[ytics label]{200\strut} -- (13.91,172.07);
	\path[xytics] (12.39,180.27) node[ytics label]{300\strut} -- (13.91,180.27);
	\path[xytics] (12.39,188.45) node[ytics label]{400\strut} -- (13.91,188.45);
	\path[xytics] (12.39,196.64) node[ytics label]{500\strut} -- (13.91,196.64);
	\path[xytics] (12.39,204.82) node[ytics label]{600\strut} -- (13.91,204.82);
\end{scope}
\begin{scope}
	\path[filledbars] (26.17,155.76) -- (26.17,196.14) -- (20.51,196.14) -- (20.51,155.76) -- cycle;
	\path[filledbars] (35.61,155.76) -- (35.61,181.75) -- (29.95,181.75) -- (29.95,155.76) -- cycle;
	\path[filledbars] (45.07,155.76) -- (45.07,175.28) -- (39.39,175.28) -- (39.39,155.76) -- cycle;
	\path[filledbars] (54.51,155.76) -- (54.51,172.42) -- (48.83,172.42) -- (48.83,155.76) -- cycle;
	\path[filledbars] (63.95,155.76) -- (63.95,170.77) -- (58.27,170.77) -- (58.27,155.76) -- cycle;
	\path[filledbars] (73.39,155.76) -- (73.39,174.05) -- (67.73,174.05) -- (67.73,155.76) -- cycle;
	\path[filledbars] (82.83,155.76) -- (82.83,170.69) -- (77.17,170.69) -- (77.17,155.76) -- cycle;
\end{scope}
\end{tikzpicture}%
}
&
        \scalebox{0.8}{\begin{tikzpicture}[x=0.8mm, y=1.1mm, inner xsep=0pt, inner ysep=0pt, outer xsep=0pt, outer ysep=0pt]

\path[xytics] (18.96,155.70) -- (94.50,155.70) -- (94.50,204.82) -- (18.96,204.82) -- cycle;

\node[outer ysep=6mm, anchor=north] at (56.73,155.70) {relax factor $\omega$\strut};
\node[rotate=90] at (9.5-3,180.26){average time per iteration (sec.)\strut};

\begin{scope}
\fontsize{9}{10}\selectfont
	\path[xytics] (28.40,154.99) node[xtics label]{0\strut}   -- (28.40,155.70);
	\path[xytics] (37.84,154.99) node[xtics label]{0.5\strut} -- (37.84,155.70);
	\path[xytics] (47.28,154.99) node[xtics label]{1.0\strut} -- (47.28,155.70);
	\path[xytics] (56.72,154.99) node[xtics label]{1.5\strut} -- (56.72,155.70);
	\path[xytics] (66.16,154.99) node[xtics label]{2.0\strut} -- (66.16,155.70);
	\path[xytics] (75.60,154.99) node[xtics label]{\texttt{auto}\strut} -- (75.60,155.70);
	\path[xytics] (85.04,154.99) node[xtics label]{\texttt{enum}\strut} -- (85.04,155.70);

	\path[xytics] (17.45,155.70) node[ytics label]{0   \strut} -- (18.96,155.70);
	\path[xytics] (17.45,167.99) node[ytics label]{0.05\strut} -- (18.96,167.99);
	\path[xytics] (17.45,180.27) node[ytics label]{0.10\strut} -- (18.96,180.27);
	\path[xytics] (17.45,192.55) node[ytics label]{0.15\strut} -- (18.96,192.55);
	\path[xytics] (17.45,204.82) node[ytics label]{0.20\strut} -- (18.96,204.82);
\end{scope}
\begin{scope}
	\path[filledbars] (31.22,155.7) -- (31.22,169.96) -- (25.56,169.96) -- (25.56,155.7) -- cycle;
	\path[filledbars] (40.66,155.7) -- (40.66,176.47) -- (35.00,176.47) -- (35.00,155.7) -- cycle;
	\path[filledbars] (50.12,155.7) -- (50.12,176.43) -- (44.44,176.43) -- (44.44,155.7) -- cycle;
	\path[filledbars] (59.56,155.7) -- (59.56,176.47) -- (53.88,176.47) -- (53.88,155.7) -- cycle;
	\path[filledbars] (69.00,155.7) -- (69.00,176.47) -- (63.32,176.47) -- (63.32,155.7) -- cycle;
	\path[filledbars] (78.44,155.7) -- (78.44,176.47) -- (72.78,176.47) -- (72.78,155.7) -- cycle;
	\path[filledbars] (87.88,155.7) -- (87.88,201.34) -- (82.22,201.34) -- (82.22,155.7) -- cycle;
\end{scope}
\end{tikzpicture}%
}
\\
\footnotesize{}(a) iteration numbers
&
\footnotesize{}(b) running time per iteration
\end{tabular}
\caption{
Performance comparison of three strategies}
\label{fig:relax-factor-det-method}
\begin{minipage}{0.8\textwidth}
$\omega=0$: non-SOR method, 
$\omega=0.5, 1, 1.5, 2$: SOR method with fixed strategy, 
¡°\texttt{enum}¡±: SOR method with enumerating strategy,
¡°\texttt{auto}¡±: SOR method with probabilistic strategy (a priori probability distribution $P(\omega=0.5)=P(\omega=1.0)=0.3$, $P(\omega=1.5)=P(\omega=2.0)=0.2$).
\end{minipage}
\end{figure}
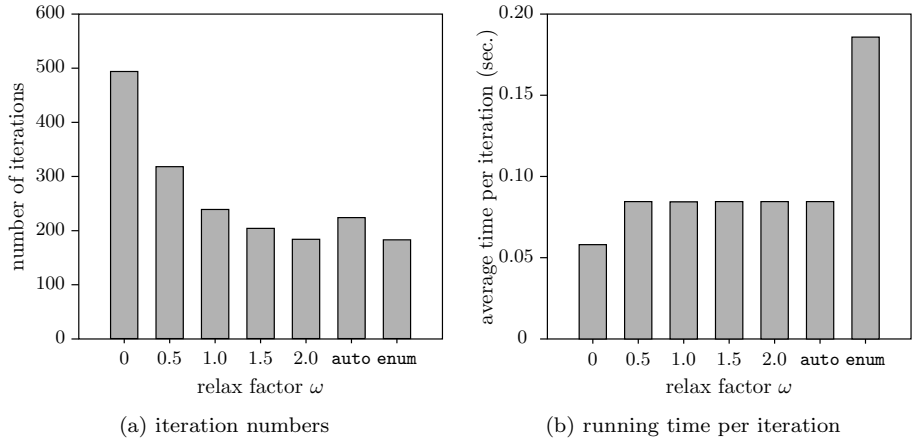

\subsection{Comparison of running time}

According to the discussion in the last subsection, we adopt three strategies (fixed, enumerating and probabilistic) on 9 data sets of graph commonly used in literature for graph drawing and test their performance in running time. The results are shown in Table~1, where \emph{a priori} probability distribution $P(\omega=0.5)=P(\omega=1.0)=0.3$, $P(\omega=1.5)=P(\omega=2.0)=0.2$ is used in probabilistic strategy. From test results, it can be shown that SOR method with fixed strategy or probabilistic strategy can decrease 30\%--40\% of running time.

\subsection{Comparison of Layout Effects}

From Algorithm 2 we know that SOR is an accelerating version of Algorithm 1 without changing the flowchart, and the layout results of these two algorithms should be the same. Notice that translation or rotation of whole placement can lead to some different layouts apparently. There are two examples: a module of protein-protein interaction network (unweighted graph) and the railway network of Chinese mainland (with edges weighted by distance between two connected vertices), drawn by using the SOR force-directed layout algorithm in Fig.~\ref{fig:layout-effect}. Compared with the layout by original non-SOR layout algorithm (not shown here and available in supplementary material), SOR method has almost the same placement as the non-SOR method regardless of translation or rotation mentioned above.

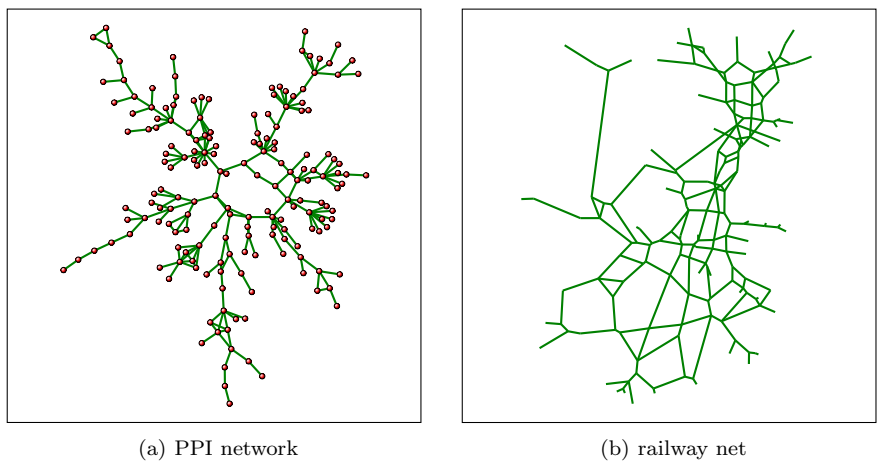
\begin{figure}[h]
\centering
\begin{tabular}{cc}
        \scalebox{1}{\begin{tikzpicture}[every edge/.style={edgestyle},scale=0.5]
\path[draw,use as bounding box] (-0.5, -0.5) rectangle (10.5, 10.5);
\path[shape=coordinate]
	(5.101, 1.901) coordinate (1)
	(6.872, 4.638) coordinate (2)
	(3.802, 3.321) coordinate (3)
	(7.797, 3.526) coordinate (4)
	(6.593, 8.224) coordinate (5)
	(5.225, 3.067) coordinate (6)
	(7.302, 3.898) coordinate (7)
	(6.613, 6.948) coordinate (8)
	(3.589, 5.699) coordinate (9)
	(4.669, 8.152) coordinate (10)
	(7.746, 3.005) coordinate (11)
	(7.532, 5.090) coordinate (12)
	(2.488, 9.112) coordinate (13)
	(4.088, 3.772) coordinate (14)
	(4.463, 8.066) coordinate (15)
	(8.266, 2.597) coordinate (16)
	(8.333, 6.380) coordinate (17)
	(7.858, 5.346) coordinate (18)
	(4.215, 6.551) coordinate (19)
	(7.799, 4.585) coordinate (20)
	(6.906, 8.314) coordinate (21)
	(2.884, 8.171) coordinate (22)
	(3.967, 8.709) coordinate (23)
	(5.460, 1.454) coordinate (24)
	(7.645, 6.213) coordinate (25)
	(4.750, 6.403) coordinate (26)
	(8.194, 6.504) coordinate (27)
	(8.137, 4.908) coordinate (28)
	(5.826, 2.264) coordinate (29)
	(6.553, 4.967) coordinate (30)
	(2.220, 9.525) coordinate (31)
	(5.037, 5.535) coordinate (32)
	(8.294, 6.136) coordinate (33)
	(0.997, 3.555) coordinate (34)
	(7.420, 9.905) coordinate (35)
	(3.159, 4.940) coordinate (36)
	(6.874, 4.381) coordinate (37)
	(2.762, 4.510) coordinate (38)
	(6.598, 6.768) coordinate (39)
	(4.993, 6.642) coordinate (40)
	(7.993, 5.254) coordinate (41)
	(7.376, 3.474) coordinate (42)
	(3.557, 3.620) coordinate (43)
	(7.495, 5.940) coordinate (44)
	(5.169, 6.179) coordinate (45)
	(8.297, 9.554) coordinate (46)
	(3.688, 6.437) coordinate (47)
	(6.584, 4.699) coordinate (48)
	(5.284, 0.966) coordinate (49)
	(5.328, 6.120) coordinate (50)
	(2.837, 7.721) coordinate (51)
	(2.647, 5.206) coordinate (52)
	(7.610, 4.655) coordinate (53)
	(1.787, 9.749) coordinate (54)
	(2.602, 8.611) coordinate (55)
	(7.116, 5.246) coordinate (56)
	(4.609, 4.210) coordinate (57)
	(3.922, 7.945) coordinate (58)
	(6.401, 7.058) coordinate (59)
	(8.295, 8.754) coordinate (60)
	(5.005, 4.740) coordinate (61)
	(6.841, 3.987) coordinate (62)
	(6.540, 8.405) coordinate (63)
	(4.760, 7.138) coordinate (64)
	(4.479, 6.712) coordinate (65)
	(6.280, 0.730) coordinate (66)
	(6.058, 6.956) coordinate (67)
	(7.554, 9.251) coordinate (68)
	(3.866, 6.854) coordinate (69)
	(4.428, 3.797) coordinate (70)
	(4.045, 5.490) coordinate (71)
	(5.308, 4.414) coordinate (72)
	(6.438, 7.706) coordinate (73)
	(3.462, 7.345) coordinate (74)
	(2.699, 4.893) coordinate (75)
	(4.672, 1.610) coordinate (76)
	(5.256, 2.476) coordinate (77)
	(6.099, 3.951) coordinate (78)
	(3.851, 7.533) coordinate (79)
	(5.426, 5.039) coordinate (80)
	(1.818, 4.071) coordinate (81)
	(7.331, 7.794) coordinate (82)
	(7.291, 5.408) coordinate (83)
	(5.920, 4.965) coordinate (84)
	(3.529, 4.844) coordinate (85)
	(5.726, 3.466) coordinate (86)
	(8.519, 6.099) coordinate (87)
	(8.003, 4.696) coordinate (88)
	(4.879, 7.067) coordinate (89)
	(4.249, 4.026) coordinate (90)
	(5.582, 2.248) coordinate (91)
	(4.264, 4.626) coordinate (92)
	(6.424, 5.118) coordinate (93)
	(3.845, 5.185) coordinate (94)
	(4.527, 7.011) coordinate (95)
	(6.630, 5.790) coordinate (96)
	(3.974, 4.564) coordinate (97)
	(5.367, 5.204) coordinate (98)
	(3.259, 7.300) coordinate (99)
	(3.716, 7.867) coordinate (100)
	(7.031, 6.166) coordinate (101)
	(6.068, 7.689) coordinate (102)
	(7.872, 4.909) coordinate (103)
	(3.720, 6.649) coordinate (104)
	(7.109, 8.179) coordinate (105)
	(5.288, 0.470) coordinate (106)
	(5.795, 6.400) coordinate (107)
	(8.040, 6.639) coordinate (108)
	(6.789, 6.400) coordinate (109)
	(7.221, 4.189) coordinate (110)
	(8.069, 9.057) coordinate (111)
	(4.522, 3.614) coordinate (112)
	(3.409, 5.585) coordinate (113)
	(7.215, 5.946) coordinate (114)
	(6.964, 5.430) coordinate (115)
	(8.282, 5.732) coordinate (116)
	(7.672, 8.803) coordinate (117)
	(6.321, 6.713) coordinate (118)
	(6.211, 6.911) coordinate (119)
	(7.349, 8.004) coordinate (120)
	(3.798, 4.784) coordinate (121)
	(2.273, 4.285) coordinate (122)
	(7.344, 9.391) coordinate (123)
	(6.779, 4.897) coordinate (124)
	(8.836, 8.775) coordinate (125)
	(5.413, 0.000) coordinate (126)
	(5.362, 1.966) coordinate (127)
	(3.911, 9.221) coordinate (128)
	(6.661, 7.368) coordinate (129)
	(3.340, 5.347) coordinate (130)
	(5.919, 4.467) coordinate (131)
	(4.801, 3.438) coordinate (132)
	(6.003, 2.972) coordinate (133)
	(3.565, 5.212) coordinate (134)
	(6.148, 6.076) coordinate (135)
	(5.891, 4.705) coordinate (136)
	(4.747, 6.683) coordinate (137)
	(3.340, 7.883) coordinate (138)
	(7.392, 6.239) coordinate (139)
	(2.054, 8.564) coordinate (140)
	(5.815, 4.036) coordinate (141)
	(6.922, 7.900) coordinate (142)
	(4.907, 7.218) coordinate (143)
	(8.405, 5.851) coordinate (144)
	(8.678, 9.133) coordinate (145)
	(6.555, 4.500) coordinate (146)
	(9.043, 6.075) coordinate (147)
	(6.157, 7.193) coordinate (148)
	(3.494, 7.581) coordinate (149)
	(5.178, 3.516) coordinate (150)
	(8.343, 3.440) coordinate (151)
	(4.552, 6.336) coordinate (152)
	(8.150, 5.132) coordinate (153)
	(7.895, 6.048) coordinate (154)
	(7.467, 6.624) coordinate (155)
	(2.130, 10.000) coordinate (156)
	(7.809, 9.180) coordinate (157)
	(6.749, 8.435) coordinate (158)
	(4.863, 8.101) coordinate (159)
	(7.368, 8.369) coordinate (160)
	(1.390, 3.832) coordinate (161)
	(4.469, 6.482) coordinate (162)
	(4.476, 5.378) coordinate (163)
	(5.413, 3.980) coordinate (164)
	(4.629, 7.610) coordinate (165)
	(2.346, 8.012) coordinate (166)
	(4.328, 7.212) coordinate (167)
	(4.046, 4.156) coordinate (168)
	(5.920, 1.126) coordinate (169)
	(4.899, 2.162) coordinate (170)
	(3.848, 6.283) coordinate (171)
	(5.176, 3.755) coordinate (172)
	(8.075, 3.071) coordinate (173)
	(7.510, 7.806) coordinate (174)
	(3.991, 8.170) coordinate (175)
	(4.305, 5.034) coordinate (176)
	(5.053, 6.850) coordinate (177)
	(2.706, 7.239) coordinate (178)
	(3.363, 8.396) coordinate (179);
\path
	(1) edge (24) edge (76) edge (77) edge (127) edge (170)
	(2) edge (30) edge (110)
	(3) edge (14)
	(4) edge (7) edge (11) edge (151) edge (173)
	(5) edge (142)
	(6) edge (77) edge (172)
	(7) edge (37) edge (42)
	(8) edge (118)
	(9) edge (71)
	(10) edge (165)
	(11) edge (173)
	(12) edge (20) edge (28) edge (41) edge (53) edge (88) edge (103) edge (115) edge (153)
	(13) edge (31) edge (55)
	(14) edge (43) edge (57) edge (70) edge (90) edge (112) edge (168)
	(15) edge (165)
	(16) edge (173)
	(17) edge (154)
	(18) edge (83)
	(19) edge (47) edge (69) edge (104) edge (137) edge (171)
	(21) edge (142)
	(22) edge (55) edge (138) edge (166)
	(23) edge (128) edge (175)
	(24) edge (49) edge (127) edge (169)
	(25) edge (108) edge (114) edge (154)
	(26) edge (137)
	(27) edge (154)
	(29) edge (77)
	(30) edge (37) edge (48) edge (84) edge (93) edge (115) edge (124) edge (146)
	(31) edge (54) edge (156)
	(32) edge (45) edge (80) edge (98) edge (163)
	(33) edge (154)
	(34) edge (161)
	(35) edge (123)
	(36) edge (38) edge (52) edge (75) edge (94)
	(38) edge (122)
	(39) edge (118)
	(40) edge (137)
	(44) edge (114)
	(45) edge (50) edge (107) edge (137)
	(46) edge (111)
	(49) edge (106)
	(51) edge (138)
	(54) edge (156)
	(55) edge (140)
	(56) edge (115)
	(57) edge (61) edge (70) edge (90) edge (112) edge (168)
	(58) edge (79) edge (175)
	(59) edge (118)
	(60) edge (117) edge (125) edge (145)
	(61) edge (98)
	(62) edge (146)
	(63) edge (142)
	(64) edge (95) edge (137) edge (165)
	(65) edge (137)
	(66) edge (169)
	(67) edge (118)
	(68) edge (117) edge (123)
	(70) edge (90) edge (112)
	(71) edge (113) edge (134) edge (163)
	(72) edge (80) edge (164) edge (172)
	(73) edge (129)
	(74) edge (79)
	(77) edge (91) edge (127) edge (170)
	(78) edge (131)
	(79) edge (99) edge (100) edge (138) edge (149) edge (167) edge (175)
	(80) edge (84)
	(81) edge (122) edge (161)
	(82) edge (142)
	(83) edge (115)
	(84) edge (98) edge (131) edge (136)
	(85) edge (94)
	(86) edge (133) edge (164)
	(87) edge (147) edge (154)
	(89) edge (95) edge (137) edge (143) edge (165)
	(92) edge (97) edge (176)
	(94) edge (130) edge (163)
	(95) edge (137) edge (143) edge (167)
	(96) edge (101) edge (115) edge (135)
	(97) edge (121) edge (176)
	(99) edge (178)
	(101) edge (114) edge (155)
	(102) edge (148)
	(105) edge (142)
	(106) edge (126)
	(107) edge (118) edge (135)
	(109) edge (114) edge (118)
	(111) edge (117)
	(114) edge (115) edge (139) edge (154)
	(116) edge (154)
	(117) edge (123) edge (157) edge (160)
	(118) edge (119) edge (129) edge (148)
	(120) edge (142)
	(121) edge (176)
	(127) edge (170)
	(129) edge (142)
	(131) edge (136) edge (141)
	(132) edge (172)
	(137) edge (143) edge (152) edge (162) edge (167) edge (177)
	(138) edge (179)
	(142) edge (158) edge (160) edge (174)
	(143) edge (165)
	(144) edge (154)
	(150) edge (164)
	(159) edge (165)
	(163) edge (176)
	(165) edge (167);
\foreach \x in {1,2,...,179}{\path[nodestyle] (\x) circle (2pt);}
\end{tikzpicture}}
&
        \scalebox{1}{\begin{tikzpicture}[every edge/.style={edgestyle},xscale=0.5,yscale=-.5]
\path[draw,use as bounding box] (-0.5, -0.5) rectangle (10.5, 10.5);
\path[shape=coordinate]
	(4.687, 0.855) coordinate (1)
	(5.509, 0.000) coordinate (2)
	(6.798, 0.229) coordinate (3)
	(5.678, 0.956) coordinate (4)
	(5.607, 0.456) coordinate (5)
	(5.206, 0.396) coordinate (6)
	(8.787, 0.762) coordinate (7)
	(5.932, 0.672) coordinate (8)
	(5.913, 0.527) coordinate (9)
	(6.817, 0.899) coordinate (10)
	(8.318, 0.054) coordinate (11)
	(8.358, 0.929) coordinate (12)
	(7.266, 1.190) coordinate (13)
	(8.087, 0.422) coordinate (14)
	(6.731, 1.527) coordinate (15)
	(5.774, 1.782) coordinate (16)
	(7.939, 0.779) coordinate (17)
	(2.223, 0.461) coordinate (18)
	(7.643, 0.623) coordinate (19)
	(7.955, 3.010) coordinate (20)
	(7.851, 1.123) coordinate (21)
	(7.447, 2.003) coordinate (22)
	(7.207, 2.069) coordinate (23)
	(3.075, 0.995) coordinate (24)
	(7.914, 1.425) coordinate (25)
	(6.983, 2.402) coordinate (26)
	(8.294, 2.527) coordinate (27)
	(3.387, 1.138) coordinate (28)
	(4.518, 3.237) coordinate (29)
	(7.867, 2.453) coordinate (30)
	(7.956, 2.425) coordinate (31)
	(4.004, 0.860) coordinate (32)
	(6.920, 3.010) coordinate (33)
	(6.305, 3.275) coordinate (34)
	(7.867, 2.639) coordinate (35)
	(6.966, 3.113) coordinate (36)
	(6.899, 3.174) coordinate (37)
	(6.512, 3.431) coordinate (38)
	(6.779, 2.961) coordinate (39)
	(6.842, 3.437) coordinate (40)
	(6.473, 3.596) coordinate (41)
	(7.656, 3.154) coordinate (42)
	(5.203, 3.864) coordinate (43)
	(6.031, 4.092) coordinate (44)
	(4.376, 3.651) coordinate (45)
	(6.189, 2.682) coordinate (46)
	(5.332, 4.397) coordinate (47)
	(4.915, 4.270) coordinate (48)
	(5.424, 4.561) coordinate (49)
	(8.109, 5.414) coordinate (50)
	(7.959, 5.232) coordinate (51)
	(6.001, 4.860) coordinate (52)
	(7.889, 5.307) coordinate (53)
	(5.467, 4.760) coordinate (54)
	(6.327, 4.909) coordinate (55)
	(7.120, 5.144) coordinate (56)
	(6.469, 4.932) coordinate (57)
	(2.939, 4.503) coordinate (58)
	(1.402, 4.543) coordinate (59)
	(3.462, 4.695) coordinate (60)
	(3.267, 4.665) coordinate (61)
	(6.990, 5.213) coordinate (62)
	(7.546, 5.273) coordinate (63)
	(6.632, 5.202) coordinate (64)
	(1.070, 4.555) coordinate (65)
	(7.574, 5.191) coordinate (66)
	(2.640, 5.062) coordinate (67)
	(4.129, 5.296) coordinate (68)
	(4.218, 5.348) coordinate (69)
	(7.105, 5.674) coordinate (70)
	(5.058, 5.417) coordinate (71)
	(6.191, 5.642) coordinate (72)
	(5.773, 5.654) coordinate (73)
	(7.045, 6.016) coordinate (74)
	(5.771, 5.502) coordinate (75)
	(5.303, 5.526) coordinate (76)
	(5.290, 5.626) coordinate (77)
	(5.706, 5.828) coordinate (78)
	(5.341, 5.763) coordinate (79)
	(6.153, 5.875) coordinate (80)
	(6.521, 5.949) coordinate (81)
	(4.749, 5.661) coordinate (82)
	(4.548, 5.730) coordinate (83)
	(4.007, 5.721) coordinate (84)
	(5.786, 6.530) coordinate (85)
	(4.147, 5.717) coordinate (86)
	(4.489, 5.777) coordinate (87)
	(5.295, 6.030) coordinate (88)
	(6.685, 6.354) coordinate (89)
	(6.010, 6.341) coordinate (90)
	(6.109, 6.594) coordinate (91)
	(3.693, 6.049) coordinate (92)
	(4.378, 6.214) coordinate (93)
	(7.024, 6.663) coordinate (94)
	(4.316, 6.273) coordinate (95)
	(6.681, 6.722) coordinate (96)
	(7.375, 7.181) coordinate (97)
	(7.306, 7.213) coordinate (98)
	(7.066, 7.047) coordinate (99)
	(6.874, 7.045) coordinate (100)
	(7.819, 7.562) coordinate (101)
	(6.097, 7.196) coordinate (102)
	(6.038, 7.264) coordinate (103)
	(6.118, 7.653) coordinate (104)
	(2.176, 6.989) coordinate (105)
	(4.464, 7.993) coordinate (106)
	(6.762, 8.312) coordinate (107)
	(5.251, 8.188) coordinate (108)
	(6.605, 8.665) coordinate (109)
	(3.689, 8.103) coordinate (110)
	(3.569, 8.035) coordinate (111)
	(1.713, 7.849) coordinate (112)
	(7.121, 8.943) coordinate (113)
	(2.127, 7.882) coordinate (114)
	(7.375, 8.688) coordinate (115)
	(2.313, 8.072) coordinate (116)
	(3.321, 9.477) coordinate (117)
	(5.420, 9.280) coordinate (118)
	(5.638, 9.335) coordinate (119)
	(4.128, 9.190) coordinate (120)
	(5.732, 9.484) coordinate (121)
	(1.556, 8.522) coordinate (122)
	(3.776, 9.601) coordinate (123)
	(4.033, 10.000) coordinate (124)
	(4.586, 9.595) coordinate (125)
	(3.824, 9.716) coordinate (126)
	(3.608, 9.676) coordinate (127)
	(4.618, 9.724) coordinate (128)
	(6.466, 1.196) coordinate (129)
	(8.228, 0.961) coordinate (130)
	(6.518, 1.440) coordinate (131)
	(7.212, 1.438) coordinate (132)
	(6.738, 1.919) coordinate (133)
	(6.516, 1.980) coordinate (134)
	(7.783, 0.949) coordinate (135)
	(7.505, 1.855) coordinate (136)
	(6.733, 2.287) coordinate (137)
	(6.802, 2.583) coordinate (138)
	(7.186, 2.402) coordinate (139)
	(7.503, 2.461) coordinate (140)
	(7.776, 2.498) coordinate (141)
	(7.150, 2.897) coordinate (142)
	(6.839, 3.284) coordinate (143)
	(6.350, 3.803) coordinate (144)
	(5.168, 3.789) coordinate (145)
	(6.204, 4.150) coordinate (146)
	(6.810, 3.816) coordinate (147)
	(6.378, 4.195) coordinate (148)
	(6.236, 4.273) coordinate (149)
	(5.401, 4.090) coordinate (150)
	(6.475, 4.342) coordinate (151)
	(5.520, 4.849) coordinate (152)
	(3.157, 5.061) coordinate (153)
	(6.513, 5.527) coordinate (154)
	(5.656, 5.767) coordinate (155)
	(5.534, 6.108) coordinate (156)
	(5.970, 6.264) coordinate (157)
	(5.025, 6.573) coordinate (158)
	(7.695, 6.980) coordinate (159)
	(6.969, 6.923) coordinate (160)
	(7.412, 7.386) coordinate (161)
	(3.785, 6.748) coordinate (162)
	(5.520, 7.082) coordinate (163)
	(3.115, 6.670) coordinate (164)
	(3.507, 7.153) coordinate (165)
	(6.394, 7.806) coordinate (166)
	(5.336, 7.894) coordinate (167)
	(7.125, 8.636) coordinate (168)
	(2.636, 8.131) coordinate (169)
	(7.289, 8.655) coordinate (170)
	(4.161, 8.853) coordinate (171)
	(3.926, 9.406) coordinate (172);
\path
	(1) edge (4)
	(2) edge (5)
	(3) edge (10)
	(4) edge (5) edge (129)
	(5) edge (8)
	(6) edge (9)
	(7) edge (12)
	(8) edge (9) edge (131)
	(10) edge (13) edge (129)
	(11) edge (14)
	(12) edge (130)
	(13) edge (130) edge (132)
	(14) edge (17)
	(15) edge (131) edge (132) edge (133)
	(16) edge (134)
	(17) edge (130) edge (135)
	(18) edge (24)
	(19) edge (135)
	(20) edge (138)
	(21) edge (25) edge (135)
	(22) edge (23) edge (136) edge (140)
	(23) edge (132) edge (133) edge (139)
	(24) edge (28)
	(25) edge (136)
	(26) edge (137) edge (138) edge (139)
	(27) edge (30)
	(28) edge (32) edge (58)
	(29) edge (145)
	(30) edge (31) edge (141)
	(33) edge (36) edge (138) edge (143)
	(34) edge (38) edge (138) edge (146)
	(35) edge (141)
	(36) edge (37) edge (142)
	(37) edge (39) edge (143)
	(38) edge (41) edge (143)
	(39) edge (46) edge (142)
	(40) edge (41) edge (143) edge (147)
	(41) edge (144)
	(42) edge (142)
	(43) edge (45) edge (145) edge (150)
	(44) edge (146) edge (150)
	(45) edge (60)
	(47) edge (48) edge (49) edge (150)
	(49) edge (54) edge (149)
	(50) edge (53)
	(51) edge (53)
	(52) edge (55) edge (73) edge (149) edge (152)
	(53) edge (63)
	(54) edge (152)
	(55) edge (57) edge (72) edge (149)
	(56) edge (62)
	(57) edge (64) edge (151)
	(58) edge (61) edge (153)
	(59) edge (65) edge (67)
	(60) edge (61) edge (86)
	(61) edge (153)
	(62) edge (63) edge (64)
	(63) edge (66)
	(64) edge (154)
	(67) edge (153)
	(68) edge (69)
	(69) edge (83)
	(70) edge (154)
	(71) edge (77) edge (82) edge (152)
	(72) edge (73) edge (80) edge (154)
	(73) edge (75) edge (155)
	(74) edge (81)
	(76) edge (77) edge (152)
	(77) edge (79)
	(78) edge (80) edge (155)
	(79) edge (82) edge (88) edge (155)
	(80) edge (81) edge (157)
	(81) edge (89) edge (154)
	(82) edge (83)
	(83) edge (87)
	(84) edge (86) edge (92) edge (153)
	(85) edge (157)
	(86) edge (87)
	(87) edge (93)
	(88) edge (156) edge (158)
	(89) edge (94) edge (96)
	(90) edge (91) edge (102) edge (157)
	(92) edge (95) edge (164)
	(93) edge (95) edge (158)
	(94) edge (159) edge (160)
	(95) edge (162)
	(96) edge (102) edge (160)
	(97) edge (98)
	(98) edge (99) edge (161)
	(99) edge (160) edge (166)
	(100) edge (160)
	(101) edge (161)
	(102) edge (103)
	(103) edge (104) edge (163)
	(104) edge (119) edge (166) edge (167)
	(105) edge (114) edge (164)
	(106) edge (110) edge (158) edge (167) edge (171)
	(107) edge (109) edge (166) edge (168)
	(108) edge (118) edge (167) edge (171)
	(110) edge (111) edge (171)
	(111) edge (165) edge (169)
	(112) edge (114)
	(113) edge (168)
	(114) edge (116)
	(115) edge (170)
	(116) edge (122) edge (169)
	(117) edge (172)
	(118) edge (119) edge (125)
	(119) edge (121)
	(120) edge (125) edge (171) edge (172)
	(123) edge (126) edge (127) edge (172)
	(124) edge (172)
	(125) edge (128)
	(129) edge (131)
	(131) edge (134)
	(132) edge (135) edge (136)
	(133) edge (134) edge (137)
	(134) edge (137)
	(137) edge (138)
	(138) edge (145)
	(139) edge (140) edge (142)
	(140) edge (141) edge (142)
	(144) edge (146)
	(146) edge (148) edge (149)
	(147) edge (148) edge (151)
	(148) edge (149) edge (151)
	(149) edge (151)
	(155) edge (156)
	(156) edge (157) edge (163)
	(158) edge (163)
	(159) edge (161)
	(161) edge (166)
	(162) edge (164) edge (165)
	(163) edge (167)
	(164) edge (165)
	(168) edge (170);
\end{tikzpicture}}
\\
\footnotesize{}(a) PPI network
& 
\footnotesize{}(b) railway net
\end{tabular}
\caption{
Graph layouts by SOR force-directed method}
\label{fig:layout-effect}
\end{figure}

\begin{table}
\centering
\tiny
\caption{Performance comparison of SOR method and original non-SOR method}
\tabcolsep 1.5pt
\centering
\begin{tabular*}{1.2\textwidth}{rrcccccccc}
\hline
       data set &       \#vertices & rel err &            & $\omega=0$ & $\omega=0.5$ & $\omega=1.0$ & $\omega=1.5$ &    auto &  enum \\
\hline
   \texttt{band-46} &         46 &   1e-06 &       num of iter &        217/100.0\% &        145/66.8\% &        136/62.7\% &        102/47.0\% &        131/60.4\% &        104/47.9\% \\
           &            &            &       running time &     1.716/100.0\% &     1.721/100.3\% &     1.607/93.7\% &     1.212/70.6\% &     1.553/90.5\% &     2.852/166.2\%  \\
\hline
   \texttt{band-62} &         62 &   1e-06 &       num of iter &        284/100.0\% &        190/66.9\% &         79/27.8\% &        125/44.0\% &         76/26.8\% &         63/22.2\% \\
           &            &            &       running time &     3.219/100.0\% &     3.249/100.9\% &     1.352/42.0\% &     2.121/65.9\% &     1.308/40.6\% &     2.492/77.4\%  \\
\hline
   \texttt{band-86} &         86 &   1e-06 &       num of iter &        188/100.0\% &        126/67.0\% &         95/50.5\% &         82/43.6\% &         90/47.9\% &         82/43.6\% \\
           &            &            &       running time &     3.177/100.0\% &     3.186/100.3\% &     2.406/75.7\% &     2.078/65.4\% &     2.298/72.4\% &     4.826/151.9\%  \\
\hline
  \texttt{band-156} &        156 &   1e-05 &       num of iter &        123/100.0\% &         83/67.5\% &         63/51.2\% &         54/43.9\% &         57/46.3\% &         57/46.3\% \\
           &            &            &       running time &     4.844/100.0\% &     4.886/100.9\% &     3.704/76.5\% &     3.189/65.8\% &     3.359/69.3\% &     7.781/160.6\%  \\
\hline
  \texttt{band-303} &        303 &   1e-05 &       num of iter &        125/100.0\% &         84/67.2\% &         64/51.2\% &         55/44.0\% &         58/46.4\% &         58/46.4\% \\
           &            &            &       running time &    15.331/100.0\%  &    15.237/99.4\%  &    11.645/76.0\%  &    10.007/65.3\%  &    10.536/68.7\%  &    24.093/157.1\%  \\
\hline
  \texttt{band-516} &        516 &   1e-05 &       num of iter &        211/100.0\% &        141/66.8\% &        107/50.7\% &         92/43.6\% &         98/46.4\% &         82/38.9\% \\
           &            &            &       running time &    65.517/100.0\%  &    64.060/97.8\%  &    48.655/74.3\%  &    41.857/63.9\%  &    44.599/68.1\%  &    84.681/129.3\%  \\
\hline
 \texttt{grid-1109} &       1109 &   1e-04 &       num of iter &        102/100.0\% &         71/69.6\% &         55/53.9\% &         49/48.0\% &         53/52.0\% &         48/47.1\% \\
           &            &            &       running time &   126.963100.0\%/ &   127.198/100.2\%  &    99.041/78.0\%  &    88.149/69.4\%  &    96.121/75.7\% &   191.292/150.7\%  \\
\hline
 \texttt{grid-1158} &       1158 &   1e-04 &       num of iter &        154/100.0\% &        104/67.5\% &         79/51.3\% &         66/42.9\% &         76/49.4\% &         64/41.6\% \\
           &            &            &       running time &   214.879/100.0\%  &   213.111/99.2\%  &   161.431/75.1\%  &   135.065/62.9\%  &   154.820/72.0\%  &   290.208/135.1\%  \\
\hline
  \texttt{net-2305} &       2305 &   1e-04 &       num of iter &         57/100.0\% &         39/68.4\% &         31/54.4\% &         27/47.4\% &         28/49.1\% &         24/42.1\% \\
           &            &            &       running time &   301.805/100.0\%  &   306.605/101.6\%  &   246.866/81.8\%  &   215.86371.5\%/ &   223.505/74.1\%  &   421.256/139.6\%  \\
\hline
        ÆœŸù &            &            &       num of iter &    100.0\% &     67.5\% &     50.4\% &     44.9\% &     47.2\% &     41.8\% \\
           &            &            &       running time &    100.0\% &    100.1\% &     74.8\% &     66.7\% &     70.2\% &    140.9\% \\
\hline
\end{tabular*}
\begin{minipage}{1\textwidth}
The percents in ``num of iter" row and ``running time" row is the ratio of SOR's value to non-SOR's and the running time is in second(s). Other notations and parameter are the same as in Fig.~\ref{fig:relax-factor-det-method}.
\end{minipage}
\end{table}

\section{Conclusion and Further Work}
\label{sec:fds-6}

Although the force-directed layout of graph is widely used, some problems still exist on both theoretical and practical aspects. Theoretically, there is no perfect way that can at least partly to ensure the convergence and estimate the convergence rate in the iterations of classical force-directed method. And practically the huge computation costs for large-scale graph, as complex network, have become the bottleneck of this method. We try to make improvements in both aspects. A new sufficient condition to guarantee the convergence of iteration process has been proposed and its corresponding convergence rate estimated. Furthermore a practical SOR-based accelerating method for the force-directed layout has been advanced and the variety of strategies to select the relax factor is discussed in detail. Numerical experiment shows that the SOR method could decrease the running time by about 30\%--40\% on average, thus speeding up the force-directed layout for graph drawing efficiently.

Further improvements to the work presented in the paper may include:

(1) We can view the relax factor as a stride length factor in Eq.~\eqref{eq10} and can try to use a variety of stride length factors corresponding to the different vertices, that is to say, $\omega =\left( {\omega_1 ,\omega_2 ,\ldots ,\omega_n } \right)\in \rvecspace{n}$ such that for each vertex $i=1,2,\ldots,n$, it is relaxed by $\tilde {x}_i^{(k+1)} \leftarrow x_i^{(k+1)} +\omega_i \left( {x_i^{(k+1)} -x_i^{(k)} } \right)$.

(2) Speedup of graph layout can facilitate the artistic of layout. For example, we can borrow the idea from genetic algorithm and view a placement as an individual. Initially a group placement is selected randomly and after a few iterations in SOR method, each initial placement will result in a better placement. The best placement among them is to be used in the subsequent process. Furthermore embedding the SOR iterations method into the genetic algorithm framework may be another good selection if the aesthetics of graph layout exceeds the speed of graph drawing.

\section*{Acknowledgements}
This work was supported by 
the National Natural Science Foundation of China (Grant No. 60603054),
the          Natural Science Foundation of Hu'nan Province, China (Grant No. 08JJ4021),
and 
the National Basic Research Program of China (Grant No. 2009CB723803).
We also thank Dr. ZHOU Ting-Ting and Mrs. DONG Yun-Yuan for their suggestions which might be of great help to improve the quality of our manuscript.
The final publication is available at \url{link.springer.com} via \url{http://dx.doi.org/10.1007/s11432-011-4208-9}.

\end{document}